\documentclass[fleqn,12pt]{llncs}

\usepackage[hidelinks]{hyperref}
\usepackage {mathpartir}
\usepackage {graphicx}
\usepackage{bigpage}
\usepackage{bcprules}
\usepackage{amsmath}
\usepackage{amssymb}
\usepackage{amsfonts}
\usepackage{amstext}
\usepackage{latexsym}
\usepackage{longtable}
\usepackage{stmaryrd}
\usepackage[svgnames]{xcolor}
\usepackage {graphicx}
\usepackage{color}

\newcommand{\ldb}{[\![}
\newcommand{\rdb}{]\!]}

\newcommand{\pplus}{\mathbin{\text{\ttfamily\upshape ++}}}
\newcommand{\id}[1]{\texttt{#1}}

\newcommand{\scong}{\mathsf{\equiv}}
\newcommand{\nameeq}{\mathsf{\equiv_N}}

\newcommand{\freenames}[1]{\mathsf{FN}(#1)}

\newcommand{\quotep}[1]{\mathsf{@} #1}

\newcommand{\substn}[2]{\{ #1 / #2 \}}

\newcommand{\meaningof}[1]{\ldb #1 \rdb}

\newcommand{\red}{\rightarrow}

\newcommand{\wbbisim}{\stackrel{\centerdot}{\approx}} %weak barbed bisimilar
\newcommand{\bc}{\mathbin{\mathbf{::=}}}
\newcommand{\bm}{\mathbin{\mathbf\mid}}
\newlength{\ltext}
\newlength{\lmath}
\newlength{\cmath}
\newlength{\rmath}
\newlength{\rtext}

\begin{document}
\def\lastname{Meredith}

\title{Meta-MeTTa: an operational semantics for MeTTa}
\titlerunning{meta-metta}

\date{2 March 2023}

\author{ Lucius Gregory Meredith\inst{1} \\
         \and 
         Ben Goertzel\inst{2}
         \and Jonathan Warrell\inst{3}
         \and Adam Vandervorst\inst{4}
}
\institute{
        {CEO, F1R3FLY.io 9336 California Ave SW, Seattle, WA 98103, USA} \\
  \email{ f1r3fly.ceo.com } \\
  \and
  {CEO, SingularityNet} \\
  \email{ ben@singularitynet.io } \\
  \and
  {Researcher, SingularityNet} \\
  \email{jonathan.warrell@singularitynet.io} \\
  \and
  \email{adam.vandervorst@singularitynet.io}
}

\maketitle              % typeset the title of the contribution

%%% ----------------------------------------------------------------------

\begin{abstract}

  We present an operational semantics for the language MeTTa.

\end{abstract}

%\keywords{Process calculi, knots, invariants}

% \begin{keyword}
% concurrency, message-passing, process calculus, reflection, program logic
% \end{keyword}

%\end{frontmatter}

%\input{mops.intro}
\section{Introduction and motivation}
We present the operational semantics for the language MeTTa. MeTTa is designed as a language in which humans and AGIs write the behavior of AGIs, being jointly developed by SingularityNet.io and F1R3FLY.io, as part of SingularityNet.io's OpenCog and Hyperon projects \cite{DBLP:conf/agi/HartG08} \cite{YT:GoertzelIklePotapovHyperon2022}. One of the principle motivations of this document is to help developers of MeTTa clients know what is a correct and compliant implementation. The document serves roughly the same function as the JVM specification or Ethereum's Yellow paper \cite{wood2014ethereum}.

\section{Towards a common language for computational dynamics}
Three of the most successful branches of scientific discourse all agree on the shape of a model adequate for expressing and effecting computation. Physics, computer science, and mathematics all use the same standard shape. A model adequate for computation comes with an algebra of states and “laws of motion.”

One paradigmatic example from physics is Hilbert spaces and the Schroedinger equation. In computer science and mathematics the algebra of states is further broken down into a monad (the free algebra of states) and an algebra of the monad recorded as some equations on the free algebra.

Computer science represents laws of motion, aka state transitions, as rewrite rules exploiting the structure of states to determine transitions to new states. Mainstream mathematics is a more recognizable generalization of physics, coding state transitions, aka behavior, via morphisms (including automorphisms) between state spaces.

But all three agree to a high degree of specificity on what ingredients go into a formal presentation adequate for effecting computation.

\subsection{Examples from computer science}
Since Milner's seminal Functions as processes paper, the gold standard for a presentation of an operational semantics is to present the algebra of states via a grammar (a monad) and a structural congruence (an algebra of the monad), and the rewrite rules in Plotkin-style SOS format \cite{DBLP:journals/mscs/Milner92} \cite{Plotkin04theorigins}.

\subsubsection{$\lambda$-calculus}

\paragraph{Algebra of States}
\begin{eqnarray*}
  Term[V] & \bc & V \\
  & \;\bm\; & \lambda V.Term[V] \\
  & \;\bm\; & \mathsf{(} Term[V] \; Term[V] \mathsf{)}
\end{eqnarray*}

The structural congruence is the usual $\alpha$-equivalence, namely that $\lambda x.M \scong \lambda y.(M\{y / x\})$ when $y$ not free in $M$.

It is evident that $Term[V]$ is a monad and imposing $\alpha$-equivalence gives an algebra of the monad.

\paragraph{Transitions}
The rewrite rule is the well know $\beta$-reduction.
\begin{mathpar}
  \inferrule* [lab=Beta]{}{((\lambda x.M)N) \to M\{N / x\}}
\end{mathpar}

\subsubsection{$\pi$-calculus}

\paragraph{Algebra of States}
\begin{eqnarray*}
  Term[N] & \bc & \mathsf{0} \\
  & \;\bm\; & \mathsf{for}\mathsf{(}N \leftarrow N\mathsf{)}Term[N] \\
  & \;\bm\; & N\mathsf{!}\mathsf{(}N\mathsf{)} \\
  & \;\bm\; & \mathsf{(}\mathsf{new}\; N\mathsf{)}Term[N] \\
  & \;\bm\; & Term[N] \; \mathsf{|}\; Term[N] \\
  & \;\bm\; & \mathsf{!}Term[N]
\end{eqnarray*}

The structural congruence is the smallest equivalence relation including $\alpha$-equivalence making $(Term[N],\;\mathsf{|}\;,\mathsf{0})$ a commutative monoid, and respecting

\begin{eqnarray*}
  \mathsf{(}\mathsf{new}\; x\mathsf{)}\mathsf{(}\mathsf{new}\; x\mathsf{)}P & \scong & \mathsf{(}\mathsf{new}\; x\mathsf{)}P \\
  \mathsf{(}\mathsf{new}\; x\mathsf{)}\mathsf{(}\mathsf{new}\; y\mathsf{)}P & \scong & \mathsf{(}\mathsf{new}\; y\mathsf{)}\mathsf{(}\mathsf{new}\; x\mathsf{)}P \\
  (\mathsf{(}\mathsf{new}\; x\mathsf{)}P)\mathsf{|}Q & \scong & \mathsf{(}\mathsf{new}\; x\mathsf{)}(P\mathsf{|}Q), x \notin \freenames{Q}
\end{eqnarray*}

Again, it is evident that $Term[N]$ is a monad and imposing the structural congruence gives an algebra of the monad.

\paragraph{Transitions}
The rewrite rules divide into a core rule, and when rewrites apply in
context.
\begin{mathpar}
  \inferrule* [lab=COMM]{}{\mathsf{for}\mathsf{(}y \leftarrow x\mathsf{)}P \mathsf{|} x\mathsf{!}\mathsf{(}z\mathsf{)} \to P\{ z / y \}} \\
  \and  
  \inferrule* [lab=PAR]{P \to P'}{(\mathsf{new}\; x)P \to (\mathsf{new}\; x)P'}
  \and  
  \inferrule* [lab=PAR]{P \to P'}{P\mathsf{|}Q \to P'\mathsf{|}Q} \\
  \and
  \inferrule* [lab=STRUCT]{P \scong P' \to Q' \scong Q}{P \to Q}
\end{mathpar}

For details see \cite{DBLP:journals/mscs/Milner92}.

\subsubsection{rho-calculus}

\paragraph{Algebra of States}
Note that the rho-calculus is different from the $\lambda$-calculus and the $\pi$-calculus because it is \emph{not} dependent on a type of variables or names. However, it does give us the opportunity to expose how ground types, such as Booleans, numeric and string operations are imported into the calculus. The calculus does depend on the notion of a $0$ process. In fact, this could be any builtin functionality. The language rholang, derived from the rho-calculus, imports all literals as \emph{processes}. Note that this is in the spirit of the $\lambda$-calculus and languages derived from it: Booleans, numbers, and strings are terms, on the level with $\lambda$ terms.

So, the parameter to the monad for the rho-calculus takes the collection of builtin processes. Naturally, for all builtin processes other than $0$ there have to be reduction rules. For brevity, we take $Z = \{0\}$.
\begin{mathpar}
\inferrule* [lab=process] {} {Term[Z] \bc Z \;\bm\; \mathsf{for}(
  Name[Z] \leftarrow Name[Z] )Term[Z] \;\bm\; Name[Z]\mathsf{!}(Term[Z]) 
  \\ \and \;\bm\; \mathsf{*}Name[Z] \;\bm\; Term[Z]\mathsf{|}Term[Z] } \\
\and \inferrule* [lab=name] {}
     {Name[Z] \bc \mathsf{@}Term[Z] }
\end{mathpar}

The structural congruence is the smallest equivalence relation including $\alpha$-equivalence making $(P,\;\mathsf{|}\;,\mathsf{0})$ a commutative monoid.

Again, it is evident that $Term[Z]$ is a monad and imposing the structural congruence gives an algebra of the monad.

\paragraph{Transitions}
The rewrite rules divide into a core rule, and when rewrites apply in
context.
\begin{mathpar}
  \inferrule* [lab=COMM] {x_{t} \;\nameeq\; x_{s}} {\mathsf{for}(y \leftarrow x_{t} )P \;\mathsf{|}\; x_{s}!(Q)
  \red P\substn{\quotep{Q}}{y}} \\
  \and
  \inferrule* [lab=PAR]{P \red P'}{P\mathsf{|}Q \red P'\mathsf{|}Q}
  \and
  \inferrule* [lab=STRUCT]{{P \;\scong\; P'} \andalso {P' \red Q'} \andalso {Q' \;\scong\; Q}}{P \red Q}
\end{mathpar}

For details see \cite{DBLP:journals/entcs/MeredithR05}.

\pagebreak
\subsubsection{The JVM}

While its complexity far exceeds the presentations above, the JVM specification respects this same shape. Here is an example from the specification of what the operation aaload does \cite{JVM7Spec}.

\begin{figure}[h!]
  \scalebox{0.75}{\includegraphics[width=\linewidth]{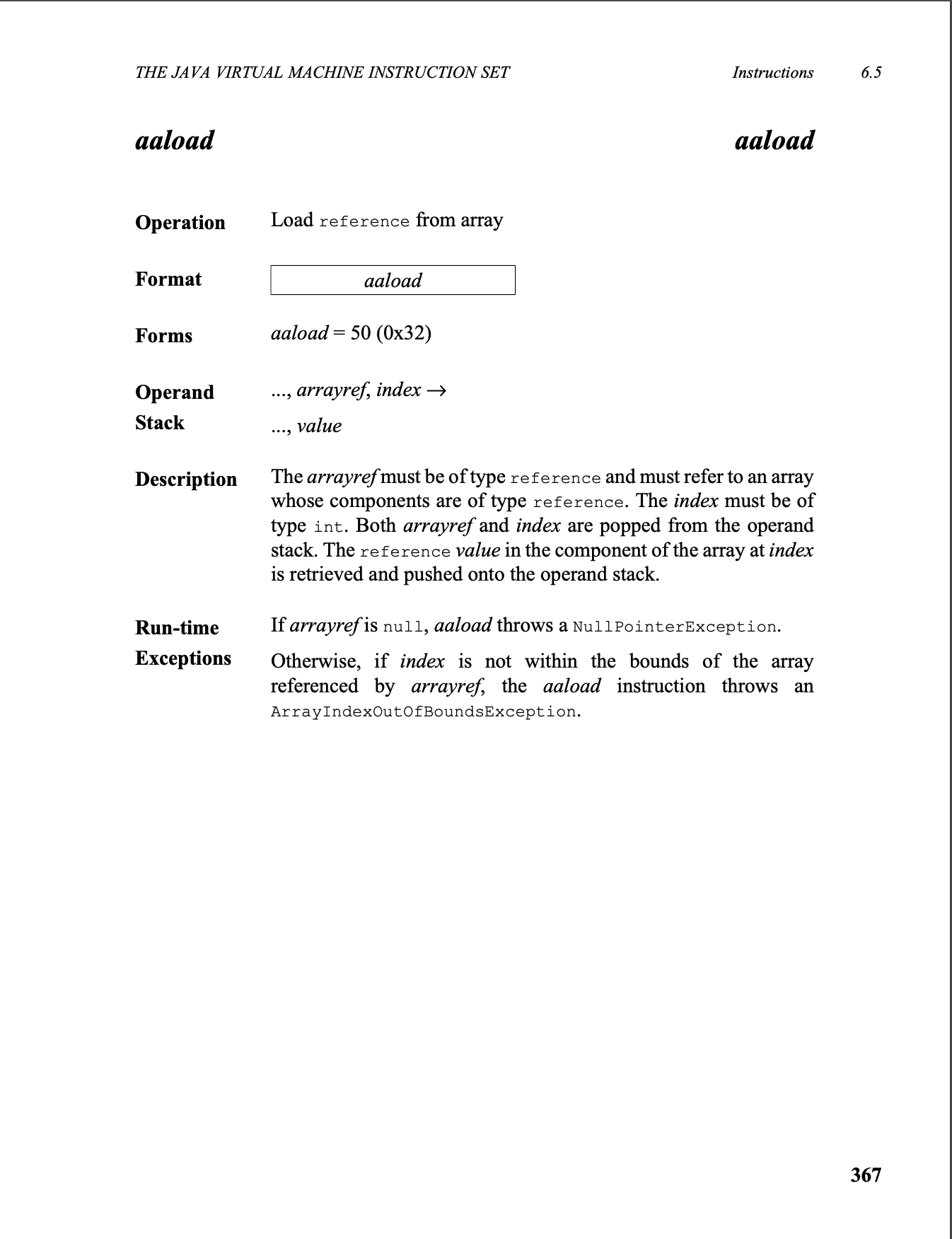}}
  \caption{AALOAD instruction specification}
  \label{fig:AALOADSpec}
\end{figure}

\paragraph{Register machines and WYSIWYG semantics}One important point about the JVM versus the previous three
examples. The first three examples are examples of WYSIWYG operational
semantics in the sense that the states \emph{are} the terms of the
calculi. In the case of the JVM the terms in the language are only
part of the state, which includes the stack, the heap, and several
registers. WYSIWYG models make static analysis dramatically
simpler. Specifically, an analyzer only has to look at terms in the
language.

This phenomenon is not restricted to the JVM. Even the famous SECD machine is not strictly WYSIWYG  \cite{DBLP:conf/sac/Buday15}. In fact, the register based states are not strictly a monad, but contain a monad.

\section{A presentation of the semantics of MeTTa}

A presentation of the semantics of MeTTa must therefore provide a monad describing the algebra of states, a structural equivalence quotienting the algebra of states, and some rewrite rules describing state transitions. Such a description is the minimal description that meets the standard for describing models of computation.

Note that to present such a description requires at least that much expressive power in the system used to formalize the presentation. That is, the system used to present a model of computation is itself a model of computation admitting a presentation in terms of an algebra of states and some rewrites. This is why a meta-circular evaluator is a perfectly legitimate presentation. That is, a presentation of MeTTa’s semantics in MeTTa is perfectly legitimate. Meta-circular presentations are more difficult to unpack, which is why such presentations are typically eschewed, but they are admissible. In fact, a meta-circular evaluator may be the most pure form of presentation.

But, this fact has an important consequence. No model that is at least Turing complete can be “lower level” than any other.

\subsection{Rationale for such a presentation}

The rationale for such a presentation is not simply that this is the way it’s done. Instead, the benefits include

\begin{itemize}
  \item an effective (if undecidable) notion of program equality;
  \item an independent specification allowing independent implementations;
  \item meta-level computation, including type checking, model checking, macros, computational reflection, etc.
\end{itemize}

\subsubsection{Effective program equality}
Of course, the notion we are calling program equality is called bisimulation in the literature. One of the key benefits of having a notion of bisimulation explicitly spelled out is that it makes possible both by-hand and automated proofs of correctness of implementations of MeTTa. We illustrate this in a later section of the paper where we specify a compilation from MeTTa to the rho-calculus and rholang and provide both a clear statement of what it means for the compiler to be correct and a proof that it is so.

One of the motivations for providing the example to is foster within the MeTTa developer community a development methodology known as correct-by-construction. Some of the benefits of correct-by-construction include avoiding bytecode injection attacks, avoiding concurrency issues, and in general avoiding technical debt.

With a clearly spelled out operational semantics the correct-by-construction software development methodology becomes a real practice and not just an ideal to be striven for -- except when we're under deadling, and we're always under deadline!

\subsubsection{Independent implementations}
Hand in hand with being able to prove an implementation correct is the ability to support multiple independent implementations, each of which is provably in compliance with the specification. Pioneered by efforts like the JVM, this approach has been remarkably effective in modern projects, like Ethereum, where the Yellow Paper made it possible for many independent teams to implement compilant Ethereum clients. This network of clients is regularly responsible for the deployment and correct execution of millions of \$ of transactions each month.

Perhaps more salient to the MeTTa developer community, it means that no one project needs to do all the development of MeTTa clients. This spreads not only the cost of development around, but the risk. In a word, it makes MeTTa more robust against the failure of any one project or team. Correct-by-construction methodology and the tools of operational semantics dramatically enhance the already proven power of open source development to decentralize cost and risk.

In short, scaling is not just about performance and throughput; it's also about adoption. Adoption at scale is a very error prone process, as anyone familiar with the histories of a wide range of computer-based technologies will attest. Linux users, for example, will bear witness to the historical lag between device drivers for Windows and Mac versus the drivers compliant to the same specs that ran on Linux. Correct-by-construction dramatically reduces the number of errors in multiple independent implementations, and thus makes scaling through adoption a much more tractable proposition.

\subsection{Meta-level computation}
Presumably efforts by human developers to develop provably correct compilation schemes from MeTTa to other computational models are just the first generation of intelligences that will seek to do so. Over time the hope is that many different kinds of intelligences will model, and hence make amenable to adaptation, MeTTa's model of computation. In particular, AGI's will seek to do the same and at dramatically different scales and timeframes. 

\subsection{MeTTa Operational Semantics}
The complexity of MeTTa's operational semantics is somewhere between the simplicity of the $\lambda$-calculus and the enormity of the JVM. Note that MeTTa is not WYSIWYG, however, it is not such a stretch to make a version of MeTTa that is.

\subsubsection{Algebra of States}

%% Non-terminals are enclosed between $<$ and $>$. 
%% The symbols -$>$ (production),  \textbf{$|$}  (union) 
%% and \textbf{eps} (empty rule) belong to the BNF notation. 
%% All other symbols are terminals.

\paragraph{Terms}
\begin{eqnarray*}
  Term & \bc & \mathsf{(} Term \; [Term] \mathsf{)} \\
  & \;\bm\; & \mathsf{\{} Term \; [Term] \mathsf{\}} \\
  & \;\bm\; & \mathsf{(} Term{} \;\mathsf{|}\; [Receipt] \; \mathsf{.} \; [Term] \mathsf{)} \\
  & \;\bm\; & \mathsf{\{} Term \;\mathsf{|}\; [Receipt] \; \mathsf{.} \; [Term] \mathsf{\}} \\
  & \;\bm\; & Atom \\
  & \;\bm\; & () \\
  & \;\bm\; & \{\} \\
\end{eqnarray*}

We impose the equation $\mathsf{\{} \ldots, t, u, \ldots \mathsf{\}} = \mathsf{\{} \ldots, u, t, \ldots \mathsf{\}}$, making terms of this form multisets. Note that for multiset comprehensions this amounts to non-determinism in the order of the terms delivered, but they are still streams. We use $\mathsf{\{}Term\mathsf{\}}$ to denote the set of terms that are (extensionally or intensionally) defined multisets, and $\mathsf{(}Term\mathsf{)}$ to denote the set of terms that are (extensionally or intensionally) defined lists.

We assume a number of polymorphic operators, such as $\pplus$ which acts as union on multisets and append on lists and concatenation on strings, and $::$ which acts as cons on lists and the appropriate generalization for the other data types.

\subsubsection{Extensionally vs intensionally defined spaces}
We make a distinction between extensionally defined spaces and terms, where each element of the space or term has been explicitly constructed, versus intensionally defined spaces and terms where elements are defined by a rule. The latter we call comprehensions.

We adopt this design for numerous reasons:
\begin{itemize}
  \item it provides an explicit representation for bindings;
  \item it provides an explicit representation for infinite terms and spaces;
  \item it provides an explicit scope for access to remotely accessed data.
\end{itemize}

\paragraph{Explicit bindings} Comprehensions provide a superior framework for the explicit representation of bindings. For example, they significantly generalize $\mathsf{let}$ and $\mathsf{letrec}$ constructs. In particular, the generally accepted semantics for $\mathsf{let}$ and $\mathsf{letrec}$ do not extend smoothly to streams, while comprehensions were effectively made for infinitary structures, including streams.

\paragraph{Infinite terms and spaces}
In fact, since the advent of set comprehensions and continuing through SQL's $\mathsf{SELECT-FROM-WHERE}$ to Haskell's do-nation and Scala's for-comprehensions the general mechanism for describing intensionally specified collections has been shown to be a powerful abstraction mechanism. Specifically, we now know that comprehensions are really syntactic sugar for the monadic operations, which makes these an exceptionally flexible framework for representing a notion of binding across a wide range of computational phenomena.

\paragraph{Remotely accessed data} Whether accessing data from a distributed atom space, or a resource on the Internet, or calling a foreign function across a memory boundary, remotely accessed data comes with different failure modes. Data providers can be offline or otherwise inaccessible. Data can be ill-formatted or peppered with an array of challenges, from buffer or register overflows to triggering divergent computation.

Providing a scope for these failure modes has a distinct advantage for defensive computation. Beyond that, however, are issues related with fair merge of divergent behavior. The famous example from PCF was logical disjunction. An evaluation strategy for disjunction that evaluates both arguments will diverge if one of the arguments diverges. An evaluation that returns true immediately if the first argument it evaluates is true will potentially diverge less often. This generalizes to a wide range of situations involving the integration of multiple foreign sources of data. 

While there is not a one-size-fits all solution, Oleg Kiselyov has provided a natural mechanism in the monad transformer, LogicT \cite{DBLP:conf/icfp/KiselyovSFS05}. This provides a policy language for describing merge policies. It is an elegant solution that fits perfectly with comprehensions.

\paragraph{Contexts}
We use McBride's notion of the derivative of a polynomial functor to calculate $1$-holed contexts for extensionally defined terms. We use $K$ to range over $\partial Term$.

\paragraph{Term sequences}
\begin{eqnarray*}
  [Term] & \bc & \epsilon \\
  & \;\bm \; & Term \\
  & \;\bm \; & Term \; [Term]
\end{eqnarray*}

\paragraph{Bindings}
\begin{eqnarray*}
  Receipt & \bc & ReceiptLinearImpl \\
  & \;\bm \; & ReceiptRepeatedImpl \\
  & \;\bm \; & ReceiptPeekImpl
\end{eqnarray*}
\begin{eqnarray*}
  [Receipt] & \bc & Receipt \\
  & \;\bm \; & Receipt \mathsf{;}\; [Receipt]
\end{eqnarray*}
\begin{eqnarray*}
  ReceiptLinearImpl & \bc & [LinearBind] \\
  LinearBind & \bc & [Name] \; NameRemainder \; \mathsf{\leftarrow} \; AtomSource
\end{eqnarray*}
\begin{eqnarray*}
  [LinearBind] & \bc & LinearBind \\
  & \;\bm \; & LinearBind \;\mathsf{\&}\; [LinearBind]
\end{eqnarray*}
\begin{eqnarray*}
  AtomSource & \bc & Name \\
  & \;\bm \; & Name \mathsf{?!} \\
  & \;\bm \; & Name \mathsf{!?} \mathsf{(} [Term] \mathsf{)} \\
  ReceiptRepeatedImpl & \bc & [RepeatedBind] \\
  RepeatedBind & \bc & [Name] \; NameRemainder \; \mathsf{\Leftarrow} \; Atom
\end{eqnarray*}
\begin{eqnarray*}
  [RepeatedBind] & \bc & RepeatedBind \\
  & \;\bm \; & RepeatedBind \; \mathsf{\&}\; [RepeatedBind] \\
  ReceiptPeekImpl & \bc & [PeekBind] \\
  PeekBind & \bc & [Name] \; NameRemainder \; \mathsf{\leftharpoonup}\; Atom
\end{eqnarray*}
\begin{eqnarray*}
  [PeekBind] & \bc & PeekBind \\
  & \;\bm \; & PeekBind \; \mathsf{\&}\; [PeekBind]
\end{eqnarray*}
\begin{eqnarray*}
  TermRemainder & \bc & \mathsf{...} \; TermVar \\
 & \;\bm \; & \epsilon \\
  NameRemainder & \bc & \mathsf{...} \; \mathsf{@} TermVar \\
& \;\bm \; & \epsilon
\end{eqnarray*}

\paragraph{Literals and builtins}
\begin{eqnarray*}
  Atom & \bc & Ground \\
  & \;\bm \; & Builtin \\
  & \;\bm \; & Var \\
  Name & \bc & \mathsf{\_} \\
  & \;\bm \; & Var \\
  & \;\bm \; & \mathsf{@} Term
\end{eqnarray*}
\begin{eqnarray*}  
  [Name] & \bc & \epsilon \\
  & \;\bm \; & Name \\
  & \;\bm \; & Name \mathsf{,} [Name] \\
  BoolLiteral & \bc & \mathsf{true} \\
 & \;\bm \; & \mathsf{false} \\
  Ground & \bc & BoolLiteral \\
  & \;\bm \; & LongLiteral \\
  & \;\bm \; & StringLiteral \\
  & \;\bm \; & UriLiteral \\
  Builtin & \bc & \mathsf{\bc} \\
  & \;\bm \; & \mathsf{=} \\
  & \;\bm \; & \mathsf{transform} \\
  & \;\bm \; & \mathsf{addAtom} \\
  & \;\bm \; & \mathsf{remAtom} \\
%  & \;\bm \; & \mathsf{:} \\
  TermVar & \bc & \mathsf{\_} \\
  & \;\bm \; & Var \\
\end{eqnarray*}

\paragraph{States}
\begin{eqnarray*}
  State & \bc & \langle \mathsf{\{}Term\mathsf{\}} \mathsf{,} \mathsf{\{}Term\mathsf{\}} \mathsf{,} \mathsf{\{}Term\mathsf{\}} \mathsf{,} \mathsf{\{}Term\mathsf{\}} \rangle
\end{eqnarray*}

We will use $S, T, U$ to range over states and $\mathsf{i} := \pi_{1}$, $\mathsf{k} := \pi_{2}$, $\mathsf{w} := \pi_{3}$and $\mathsf{o} := \pi_{4}$ for the first, second, third, and fourth projections as accessors for the components of states. Substitutions are ranged over by $\sigma$, and as is standard, substitution application will be written postfix, e.g. $t\sigma$.

A state should be thought of as consisting of $4$ \emph{registers}:
\begin{itemize}
  \item $\mathsf{i}$ is the input register where queries are issued;
  \item $\mathsf{k}$ is the knowledge base;
  \item $\mathsf{w}$ is a workspace;
  \item $\mathsf{o}$ is the output register.
\end{itemize}

We separate the input, workspace, and output registers to allow for
coarse-graining of bisimulation. An external agent cannot necessarily observe the transitions related to the workspace.

\paragraph{State contexts}
We lift term contexts to states. Thus, if $t = K[u]$ and $u \in S_{r}$ for $r \in \{ \mathsf{i}, \mathsf{k}, \mathsf{w}, \mathsf{o}\}$, then we write $K[S]$ for the state in which $t$ replaces $u$ in $S_{r}$.

\subsubsection{Rewrite Rules}

\begin{mathpar}
  \inferrule* [lab=Query]{\sigma_{i} = \mathsf{unify}(t',t_{i}), k = \mathsf{\{} (\mathsf{=}\; t_{1} \; u_{1}), \ldots, (\mathsf{=}\; t_{n} \; u_{n}) \mathsf{\}} \pplus k', \mathsf{insensitive}(t',k')}{\langle \mathsf{\{} K[t'] \mathsf{\}} \pplus i, k, w, o \rangle \to \langle i, k, \mathsf{\{} K[u_{1}\sigma_{1}] \mathsf{\}} \pplus\; \ldots\; \pplus \mathsf{\{} K[u_{n}\sigma_{n}] \mathsf{\}} \pplus w, o \rangle} \\
  \and
  \inferrule* [lab=Chain]{\sigma_{i} = \mathsf{unify}(u,t_{i}), k = \mathsf{\{} (\mathsf{=}\; t_{1} \; u_{1}), \ldots, (\mathsf{=}\; t_{n} \; u_{n}) \mathsf{\}} \pplus k', \mathsf{insensitive}(u,k')}{\langle i, k, \mathsf{\{} K[u] \mathsf{\}} \pplus w, o \rangle \to \langle i, k, \mathsf{\{} K[u_{1}\sigma_{1}] \mathsf{\}} \pplus\; \ldots\; \pplus \mathsf{\{} K[u_{n}\sigma_{n}] \mathsf{\}} \pplus w, o \rangle} \\
  \and
  \inferrule* [lab=Transform]{\sigma_{i} = \mathsf{unify}(t,t_{i}), k = \mathsf{\{} K_{1}[t_{1}], \ldots, K_{n}[t_{n}] \mathsf{\}} \pplus k',\mathsf{insensitive}(t,k')}{\langle \mathsf{\{} \mathsf{(}\mathsf{transform}\; t \; u\mathsf{)} \mathsf{\}} \pplus i, k, w, o \rangle \to \langle i, k, \mathsf{\{} K_{1}[u\sigma_{1}] \mathsf{\}} \pplus\; \ldots\; \pplus \mathsf{\{} K_{n}[u\sigma_{n}] \mathsf{\}} \pplus w, o \rangle} \\
  \and
  \inferrule* [lab=AddAtom1]{}{\langle \mathsf{\{} \mathsf{(} \mathsf{addAtom}\; t\mathsf{)}\mathsf{\}}  \pplus i, k, w, o \rangle \to \langle i, k \pplus \mathsf{\{} t\mathsf{\}}, w, \mathsf{\{} \mathsf{()}\mathsf{\}} \pplus o \rangle} \\
  \and
  \inferrule* [lab=AddAtom2]{\langle i_{1}, k_{1}, w_{1}, o_{1}\rangle \to \langle i_{2}, k_{2}, w_{2}, o_{2} \rangle, k_{3} = \mathsf{\{} \mathsf{(} \mathsf{addAtom}\; t\mathsf{)}\mathsf{\}} \pplus k_{1}}{\langle i_{1}, k_{3}, w_{1}, o_{1}\rangle \to \langle i_{2}, \mathsf{\{} \mathsf{(} \mathsf{addAtom}\; t\mathsf{)}, t\mathsf{\}} \pplus k_{2}, w_{2}, \mathsf{\{} \mathsf{()}\mathsf{\}} \pplus o_{2} \rangle} \\
  \inferrule* [lab=RemAtom1]{}{\langle \mathsf{\{} \mathsf{(} \mathsf{remAtom}\; t\mathsf{)}\mathsf{\}}  \pplus i, \mathsf{\{} t \mathsf{\}} \pplus k, w, o \rangle \to \langle i, k, w, \mathsf{\{} \mathsf{()}\mathsf{\}} \pplus o \rangle} \\
  \and
  \inferrule* [lab=RemAtom2]{\langle i_{1}, k_{1}, w_{1}, o_{1}\rangle \to \langle i_{2}, k_{2}, w_{2}, o_{2} \rangle, k_{3} = \mathsf{\{} \mathsf{(} \mathsf{remAtom}\; t\mathsf{)} \mathsf{\}} \pplus \mathsf{\{} t \mathsf{\}} \pplus k_{1}}{\langle i_{1}, k_{3}, w_{1}, o_{1}\rangle \to \langle i_{2}, \mathsf{\{} \mathsf{(} \mathsf{remAtom}\; t\mathsf{)}\mathsf{\}} \pplus k_{2}, w_{2}, \mathsf{\{} \mathsf{()}\mathsf{\}} \pplus o_{2} \rangle} \\
  \and
  \inferrule* [lab=Output]{\mathsf{insensitive}(u,k)}{\langle i, k, \mathsf{\{} u \mathsf{\}} \pplus w, o \rangle \to \langle i, k, w, \mathsf{\{} u \mathsf{\}} \pplus o \rangle} \\
\end{mathpar}

Where $\mathsf{insensitive}(t,k)$ means that $\mathsf{(} \mathsf{=}\; t'\; u \mathsf{)} \in k \Rightarrow \neg \mathsf{unify}(t,t')$.

\section{Ground literals and builtins}
As with all practical programming languages, MeTTa hosts a number of computational entities and operations that are already defined on the vast majority of platforms on which an implementation of the language may be written and/or run. Here we describe the ground literals and builtin operations that every compliant MeTTa operation must provide.
\subsection{Ground literals}
As the grammar spells out every compliant implementation of MeTTa must provide:

\begin{itemize}
  \item Booleans;
  \item signed and unsigned 64bit integers;
  \item 64bit floating point;
  \item strings
\end{itemize}

\subsection{Polymorphic operations}
Every compliant implementation of the MeTTa client must provide the following polymorphic operations:

\begin{itemize}
  \item $* : A \times A \rightarrow A$ for A ranging over Booleans, integers, floating point;
  \item $+ : A \times A \rightarrow A$ for A ranging over Booleans, integers, floating point, and strings;
\end{itemize}

\subsection{Transition rules}
\begin{mathpar}
  \inferrule* [lab=BoolAdd1]{}{\langle \mathsf{\{} \mathsf{(} \mathsf{+}\; b_{1} \; b_{2} \mathsf{)}\mathsf{\}}  \pplus i, k, w, o \rangle \to \langle i, k, w, \mathsf{\{} b_{1}\mathsf{||} b_{2} \mathsf{\}} \pplus o \rangle} \\
  \and
  \inferrule* [lab=BoolAdd2]{w = \mathsf{\{} \mathsf{(} \mathsf{+}\; b_{1} \; b_{2} \mathsf{)}\mathsf{\}} \pplus w'}{\langle i, k, w, o \rangle \to \langle i, k, w', \mathsf{\{} b_{1}\mathsf{||} b_{2} \mathsf{\}} \pplus o \rangle} \\
  \and
  \inferrule* [lab=BoolMult1]{}{\langle \mathsf{\{} \mathsf{(} \mathsf{*}\; b_{1} \; b_{2} \mathsf{)}\mathsf{\}}  \pplus i, k, w, o \rangle \to \langle i, k, w, \mathsf{\{} b_{1}\mathsf{\&} b_{2} \mathsf{\}} \pplus o \rangle} \\
  \and
  \inferrule* [lab=BoolMult2]{w = \mathsf{\{} \mathsf{(} \mathsf{*}\; b_{1} \; b_{2} \mathsf{)}\mathsf{\}} \pplus w'}{\langle i, k, w, o \rangle \to \langle i, k, w', \mathsf{\{} b_{1}\mathsf{\&} b_{2} \mathsf{\}} \pplus o \rangle} \\
  \and
  \inferrule* [lab=NumAdd1]{}{\langle \mathsf{\{} \mathsf{(} \mathsf{+}\; n_{1} \; n_{2} \mathsf{)}\mathsf{\}}  \pplus i, k, w, o \rangle \to \langle i, k, w, \mathsf{\{} n_{1}\mathsf{+} n_{2} \mathsf{\}} \pplus o \rangle} \\
  \and
  \inferrule* [lab=NumAdd2]{w = \mathsf{\{} \mathsf{(} \mathsf{+}\; n_{1} \; n_{2} \mathsf{)}\mathsf{\}} \pplus w'}{\langle i, k, w, o \rangle \to \langle i, k, w', \mathsf{\{} n_{1}\mathsf{+} n_{2} \mathsf{\}} \pplus o \rangle} \\
  \inferrule* [lab=NumMult1]{}{\langle \mathsf{\{} \mathsf{(} \mathsf{*}\; n_{1} \; n_{2} \mathsf{)}\mathsf{\}}  \pplus i, k, w, o \rangle \to \langle i, k, w, \mathsf{\{} n_{1}\mathsf{*} n_{2} \mathsf{\}} \pplus o \rangle} \\
  \and
  \inferrule* [lab=NumMult2]{w = \mathsf{\{} \mathsf{(} \mathsf{*}\; n_{1} \; n_{2} \mathsf{)}\mathsf{\}} \pplus w'}{\langle i, k, w, o \rangle \to \langle i, k, w', \mathsf{\{} n_{1}\mathsf{*} n_{2} \mathsf{\}} \pplus o \rangle} \\
  \and
  \inferrule* [lab=StrAdd1]{}{\langle \mathsf{\{} \mathsf{(} \mathsf{+}\; s_{1} \; s_{2} \mathsf{)}\mathsf{\}}  \pplus i, k, w, o \rangle \to \langle i, k, w, \mathsf{\{} s_{1}\mathsf{+} s_{2} \mathsf{\}} \pplus o \rangle} \\
  \and
  \inferrule* [lab=StrAdd2]{w = \mathsf{\{} \mathsf{(} \mathsf{+}\; s_{1} \; s_{2} \mathsf{)}\mathsf{\}} \pplus w'}{\langle i, k, w, o \rangle \to \langle i, k, w', \mathsf{\{} s_{1}\mathsf{+} s_{2} \mathsf{\}} \pplus o \rangle}
\end{mathpar}

\section{Bisimulation}
Since the operational semantics is expressed as a transition system we recover a notion of bisimulation. There are two possible ways to generate the notion of bisimulation in this context. One uses the Leifer-Milner-Sewell approach of deriving a bisimulation from the rewrite rules \cite{DBLP:conf/concur/LeiferM00}. However, the technical apparatus is very heavy to work with. The other is to adapt barbed bisimulation developed for the asynchronous $\pi$-calculus to this setting \cite{DBLP:books/daglib/0004377}.

The reason we need to use some care in developing the notion of bisimulation is that there are substitutions being generated and applied in many of the rules. So, a single label will not suffice. However, taking a query in the input space as a barb will. This notion of barbed bisimulation will provide a means of evaluating the correctness of compilation schemes to other languages. We illustrate this idea in the section on compiling MeTTa to the rho-calculus.

\section{The cost of transitions}
\subsection{Network access tokens}
If you’re reading this, chances are that you know what an Internet-facing API is, and why it might need to be protected from denial of service attacks. But, just in case you’re one of the “normies”  that don’t know what these terms refer to, let’s you, me, Sherman, and Mr. Peabody all take a trip in the WayBack Machine way back to 2005. 

In those days there was still a naivete about the infinite potential of free and open information. QAnon, deep fakes, ChatGPT and other intimations that the Internet might just be the modern equivalent of the Tower of Babel were not yet even a gleam in their inventors’ eyes. Companies would regularly set up network services that anyone with an Internet connection could access, from anywhere in the world (dubbed Internet-facing). Such services were accessed by sending requests in a particular, well defined format (deriving from the software term application program interface, or API) to an Internet address served by machines in the network service the organization had set up. 

It was quickly discovered that such Internet-facing APIs were vulnerable to attack. If a single bad actor sent thousands or millions of requests to the service, or a botnet of millions sent a few requests each to the service, it was possible for the service to become bogged down and unresponsive to legitimate requests. Now, in reality, all this was discovered long before 2005. But, by 2005 a practice for dealing with this kind of attack was more or less well established. 

The solution is simple. The network proprietor issues a digital token. A request with a given token embedded in it is honored, up to some number of requests per token. This practice is less onerous and costly than having to issue and maintain authorization credentials for login challenges. Many, many companies do this and have done this for the better part of two decades. Not just software or digital service companies like Google and Microsoft issue tokens like this,  Other companies, such as media companies like The New York Times, and The Guardian, also employ this practice. (The hyperlinks above are to their token distribution pages.) The practice is ubiquitous and well accepted. It is intrinsic to the functionality of an open network such as the Web.

Also, it is important to note that many of these services allow for storage of digital content on the networks provided by these services. However, bad actors can still abuse the services by repeatedly uploading illegal content (like child pornography, copyrighted material or even nuclear secrets). So, an entity offering Internet-enabled services must reserve the right to invalidate these tokens in case they discover they are being abused in this or other ways. These utility tokens are essential to comply with a whole host of very good laws.

\subsection{Ethereum’s big idea}
Satoshi’s discovery of a new class of economically secured, leaderless distributed consensus protocols, embodied in proof-of-work but also, elsewhere, embodied in proof-of-stake and other consensus algorithms, was a pretty good idea. It led to the Bitcoin network. Buterin’s suggestion that Satoshi’s consensus be applied to the state of a virtual machine instead of a ledger was a really good idea, and led to the Ethereum network. It creates a distributed computer that runs everywhere and nowhere in particular. Less poetically, every node in the network is running a copy of the virtual machine and the consensus protocol ensures that all the copies agree on the state of the virtual machine.

Like the Internet-facing APIs launched all throughout the 00’s and beyond, Ethereum’s distributed computer is accessible to anyone with an Internet connection. And, as such, without protection would be vulnerable to denial of service  attacks. In fact, it’s potentially even more vulnerable because a request to the Ethereum distributed computer is a piece of code. This code could, in principle, run forever, or take up infinite storage space. Vitalik’s clever idea, building on the established practice of network access tokens, is to require tokens for each computational or storage step to prevent such abuses.

\subsection{MeTTa effort objects}
MeTTa takes the same approach. Transitions in the operational semantics cost a computational resource (effort objects, or EOs, for short) that are ``purchased'' with tokens. This section reprises the operational semantics with the cost of each step spelled out in terms of the structure of EOs.

\subsubsection{Resource-bounded Rewrite Rules}

We assume a polymorphic cost function $\mathsf{\#}$ taking values in the domain of EOs. We assume the domain of EOs supports notions of $+$ and $-$ making it a \emph{commutative} group. We use EOs$_{\bot}$ to denote EOs extended with $\bot$ to indicate no EOs assigned. We assume a term representation of elements of EOs.

Additionally, we adopt the notation $t_{\chi(p,eos)}$ to indicate a pair consiting of a term and an element of the domain of EOs$_{\bot}$, signed with a private key $p$; and assume we can lift the usual functions (e.g., $\mathsf{unify}$ and $\mathsf{insensitive}$) to these signed terms either by projecting to the term component, or by some other more sophisticated mechanism.

Furthermore, we expand the state to include a fifth register consisting of an element of $\{ Term \}$, written $eos$, which in turn contains terms of the form $(h(p)\; eo)$, where $h(p)$ is a function of a private key $p$, and $eo$ is an element of the domain EOs.

\begin{mathpar}
  \inferrule* [lab=Query]{k = \mathsf{\{} (\mathsf{=}\; t_{1} \; u_{1}), \ldots, (\mathsf{=}\; t_{n} \; u_{n}) \mathsf{\}} \pplus k', \and \mathsf{insensitive}(t',k') \\ \sigma_{i} = \mathsf{unify}(t',t_{i}), \and w' = \mathsf{\{} (K[u_{1}\sigma_{1}])_{\chi(p,\bot)} \mathsf{\}} \pplus\; \ldots\; \pplus \mathsf{\{} (K[u_{n}\sigma_{n}])_{\chi(p,\bot)} \mathsf{\}} \pplus w \\ c = \Sigma_{i}\mathsf{\#}(\sigma_{i}) + \Sigma_{i}\mathsf{\#}(u_{i}\sigma_{i}), \and ((e' + e ) - c) > 0 \\ eos = \{(h(p)\; e')\} \pplus eos', \and eos'' = \{(h(p)\; ((e' + e ) - c))\} \pplus eos' }{\langle \mathsf{\{} K[t'_{\chi(p,e)}] \mathsf{\}} \pplus i, k, w, o; eos \rangle \xrightarrow{c} \langle i, k, w', o; eos'' \rangle} \\
  \and
  \inferrule* [lab=Chain]{k = \mathsf{\{} (\mathsf{=}\; t_{1} \; u_{1}), \ldots, (\mathsf{=}\; t_{n} \; u_{n}) \mathsf{\}} \pplus k',\and \mathsf{insensitive}(u,k') \\ \sigma_{i} = \mathsf{unify}(u,t_{i}), \and w' = \mathsf{\{} (K[u_{1}\sigma_{1}])_{\chi(p,\bot)} \mathsf{\}} \pplus\; \ldots\; \pplus \mathsf{\{} (K[u_{n}\sigma_{n}])_{\chi(p,\bot)} \mathsf{\}} \pplus w \\ c = \Sigma_{i}\mathsf{\#}(\sigma_{i}) + \Sigma_{i}\mathsf{\#}(u_{i}\sigma_{i}), \and (e - c) > 0 \\ eos = \{(h(p)\; e)\} \pplus eos', \and eos'' = \{(h(p)\; (e - c))\} \pplus eos'}{\langle i, k, \mathsf{\{} K[u]_{\chi(p,\bot)} \mathsf{\}} \pplus w, o; eos \rangle \xrightarrow{c} \langle i, k,  w', o; eos'' \rangle} \\
  \and
  \inferrule* [lab=Transform]{k = \mathsf{\{} t_{1}, \ldots, t_{n} \mathsf{\}} \pplus k',\and \mathsf{insensitive}(t,k')\\ \sigma_{i} = \mathsf{unify}(t,t_{i}), \and w' = \mathsf{\{} (K_{1}[u\sigma_{1}])_{\chi(p,\bot)} \mathsf{\}} \pplus\; \ldots\; \pplus \mathsf{\{} (K_{n}[u\sigma_{n})]_{\chi(p,\bot)} \mathsf{\}} \pplus w \\ c = \Sigma_{i}\mathsf{\#}(\sigma_{i}) + \Sigma_{i}\mathsf{\#}(u_{i}\sigma_{i}), \and ((e' + e ) - c) > 0 \\ eos = \{(h(p)\; e')\} \pplus eos', \and eos'' = \{(h(p)\; ((e' + e ) - c))\} \pplus eos'}{\langle \mathsf{\{} \mathsf{(}\mathsf{transform}\; t \; u\mathsf{)}_{\chi(p,e)} \mathsf{\}} \pplus i, k, w, o; eos \rangle \xrightarrow{c} \langle i, k, w', o; eos'' \rangle} \\
  \and
  \inferrule* [lab=AddAtom1]{((e + e') - \mathsf{\#}(t)) > 0 \\ eos = \{(h(p)\; e')\} \pplus eos', \and eos'' = \{(h(p)\; ((e + e') - \mathsf{\#}(t))) \} \pplus eos'}{\langle \mathsf{\{} \mathsf{(} \mathsf{addAtom}\; t\mathsf{)}_{\chi(p,e)} \mathsf{\}}  \pplus i, k, o; eos \rangle \xrightarrow{\mathsf{\#}(t)} \langle i, k \pplus \mathsf{\{} t\mathsf{\}}, w, \mathsf{\{} \mathsf{()}\mathsf{\}} \pplus o; eos'' \rangle} \\
  \and
  \inferrule* [lab=AddAtom2]{\langle i_{1}, k_{1}, w_{1}, o_{1}; eos_{1}\rangle \xrightarrow{c} \langle i_{2}, k_{2}, w_{2}, o_{2}; eos_{2} \rangle, k_{3} = \mathsf{\{} \mathsf{(} \mathsf{addAtom}\; t\mathsf{)}_{\chi(p,\bot)}\mathsf{\}} \pplus k_{1} \\ (e - \mathsf{\#}(t)) > 0 \\ eos_{1}' = \{(h(p)\; e)\} \pplus eos_{1}, \and eos_{2}' = \{(h(p)\; (e - \mathsf{\#}(t))) \} \pplus eos_{2}}{\langle i_{1}, k_{3}, w_{1}, o_{1}; eos_{1}'\rangle \xrightarrow{\mathsf{\#}(t)} \langle i_{2}, \mathsf{\{} \mathsf{(} \mathsf{addAtom}\; t\mathsf{)}_{\chi(p,\bot)}, t_{\chi(p,\bot)}\mathsf{\}} \pplus k_{2}, w_{2}, \mathsf{\{} \mathsf{()}\mathsf{\}} \pplus o_{2}; eos_{2}' \rangle} \\
  \inferrule* [lab=RemAtom1]{((e + e') - \mathsf{\#}(t)) > 0 \\ eos = \{(h(p)\; e')\} \pplus eos', \and eos'' = \{(h(p)\; ((e + e') - \mathsf{\#}(t))) \} \pplus eos'}{\langle \mathsf{\{} \mathsf{(} \mathsf{remAtom}\; t\mathsf{)}_{\chi(p,e)} \mathsf{\}}  \pplus i, \mathsf{\{} t \mathsf{\}} \pplus k, w, o; eos \rangle \xrightarrow{\mathsf{\#}(t)} \langle i, k,  \mathsf{\{} w, \mathsf{()}\mathsf{\}} \pplus o; eos'' \rangle} \\
  \and
  \inferrule* [lab=RemAtom2]{\langle i_{1}, k_{1}, w_{1}, o_{1}; eos_{1}\rangle \xrightarrow{c} \langle i_{2}, k_{2}, w_{2}, o_{2}; eos_{2} \rangle, k_{3} = \mathsf{\{} \mathsf{(} \mathsf{remAtom}\; t\mathsf{)}_{\chi(p,\bot)} \mathsf{\}} \pplus \mathsf{\{} t \mathsf{\}} \pplus k_{1} \\ (e - \mathsf{\#}(t)) > 0 \\ eos_{1}' = \{(h(p)\; e)\} \pplus eos_{1}, \and eos_{2}' = \{(h(p)\; (e - \mathsf{\#}(t))) \} \pplus eos_{2}}{\langle i_{1}, k_{3}, w_{1}, o_{1}; eos_{1}\rangle \xrightarrow{\mathsf{\#}(t)} \langle i_{2}, \mathsf{\{} \mathsf{(} \mathsf{remAtom}\; t\mathsf{)}_{\chi(p,\bot)}\mathsf{\}} \pplus k_{2}, w_{2}, \mathsf{\{} \mathsf{()}\mathsf{\}} \pplus o_{2}; eos_2 \rangle} \\
  \and
  \inferrule* [lab=Output]{\mathsf{insensitive}(u,k), \\ (e - \mathsf{\#}(u)) > 0 \\ eos = \{(h(p)\; e)\} \pplus eos', \and eos'' = \{(h(p)\; (e - \mathsf{\#}(u))) \} \pplus eos'}{\langle i, k, \mathsf{\{} u_{\chi(p,\bot)} \mathsf{\}} \pplus w, o; eos \rangle \xrightarrow{\mathsf{\#}(u)} \langle i, k, w, \mathsf{\{} u \mathsf{\}} \pplus o; eos'' \rangle}
\end{mathpar}

Likewise, the builtin operations must come with a cost.

\begin{mathpar}
  \inferrule* [lab=BoolAdd1]{((e + e') - {\#}(b_{1}) + \mathsf{\#}(b_{2})) > 0 \\ eos = \{(h(p)\; e')\} \pplus eos', \and eos'' = \{(h(p)\; ((e + e') - {\#}(b_{1}) + \mathsf{\#}(b_{2}))) \} \pplus eos'}{\langle \mathsf{\{} \mathsf{(} \mathsf{+}\; b_{1} \; b_{2} \mathsf{)}_{\chi(p,e)}\mathsf{\}}  \pplus i, k, w, o; eos \rangle \xrightarrow{{\#}(b_{1}) + \mathsf{\#}(b_{2})} \langle i, k, w, \mathsf{\{} b_{1}\mathsf{||} b_{2} \mathsf{\}} \pplus o; eos'' \rangle} \\
  \and
  \inferrule* [lab=BoolAdd2]{w = \mathsf{\{} \mathsf{(} \mathsf{+}\; b_{1} \; b_{2} \mathsf{)}_{\chi(p,e)}\mathsf{\}} \pplus w' \\ eos = \{(h(p)\; e)\} \pplus eos', \and eos'' = \{(h(p)\; (e - \mathsf{\#}(b_{1}) + \mathsf{\#}(b_{2}))) \} \pplus eos'}{\langle i, k, w, o; eos \rangle \xrightarrow{{\#}(b_{1}) + \mathsf{\#}(b_{2})} \langle i, k, w', \mathsf{\{} b_{1}\mathsf{||} b_{2} \mathsf{\}} \pplus o; eos'' \rangle} \\
  \and
  \inferrule* [lab=BoolMult1]{((e + e') - \mathsf{\#}(b_{1}) + \mathsf{\#}(b_{2})) > 0 \\ eos = \{(h(p)\; e')\} \pplus eos', \and eos'' = \{(h(p)\; ((e + e') - \mathsf{\#}(b_{1}) + \mathsf{\#}(b_{2}))) \} \pplus eos'}{\langle \mathsf{\{} \mathsf{(} \mathsf{*}\; b_{1} \; b_{2} \mathsf{)}_{\chi(p,e)}\mathsf{\}}  \pplus i, k, w, o; eos \rangle \xrightarrow{{\#}(b_{1}) + \mathsf{\#}(b_{2})} \langle i, k, w, \mathsf{\{} b_{1}\mathsf{\&} b_{2} \mathsf{\}} \pplus o; eos'' \rangle} \\
  \and
  \inferrule* [lab=BoolMult2]{w = \mathsf{\{} \mathsf{(} \mathsf{*}\; b_{1} \; b_{2} \mathsf{)}_{\chi(p,e)}\mathsf{\}} \pplus w' \\ eos = \{(h(p)\; e)\} \pplus eos', \and eos'' = \{(h(p)\; (e - \mathsf{\#}(b_{1}) + \mathsf{\#}(b_{2}))) \} \pplus eos'}{\langle i, k, w, o; eos \rangle \xrightarrow{{\#}(b_{1}) + \mathsf{\#}(b_{2})} \langle i, k, w, \mathsf{\{} b_{1}\mathsf{\&} b_{2} \mathsf{\}} \pplus o; eos'' \rangle} \\
\end{mathpar}
\begin{mathpar}
  \inferrule* [lab=NumAdd1]{((e + e') - \mathsf{\#}(b_{1}) + \mathsf{\#}(b_{2})) > 0 \\ eos = \{(h(p)\; e')\} \pplus eos', \and eos'' = \{(h(p)\; ((e + e') - \mathsf{\#}(n_{1}) + \mathsf{\#}(n_{2}))) \} \pplus eos'}{\langle \mathsf{\{} \mathsf{(} \mathsf{+}\; n_{1} \; n_{2} \mathsf{)}_{\chi(p,e)}\mathsf{\}}  \pplus i, k, w, o; eos \rangle \xrightarrow{{\#}(n_{1}) + \mathsf{\#}(n_{2})} \langle i, k, w, \mathsf{\{} n_{1}\mathsf{+} n_{2} \mathsf{\}} \pplus o; eos'' \rangle} \\
  \and
  \inferrule* [lab=NumAdd2]{w = \mathsf{\{} \mathsf{(} \mathsf{+}\; n_{1} \; n_{2} \mathsf{)}_{\chi(p,e)}\mathsf{\}} \pplus w' \\ eos = \{(h(p)\; e)\} \pplus eos', \and eos'' = \{(h(p)\; (e - \mathsf{\#}(n_{1}) + \mathsf{\#}(n_{2}))) \} \pplus eos'}{\langle i, k, w, o; eos \rangle \xrightarrow{{\#}(n_{1}) + \mathsf{\#}(n_{2})} \langle i, k, w', \mathsf{\{} n_{1}\mathsf{+} n_{2} \mathsf{\}} \pplus o; eos'' \rangle} \\
  \inferrule* [lab=NumMult1]{((e + e') - \mathsf{\#}(n_{1}) + \mathsf{\#}(n_{2})) > 0 \\ eos = \{(h(p)\; e')\} \pplus eos', \and eos'' = \{(h(p)\; ((e + e') - \mathsf{\#}(n_{1}) + \mathsf{\#}(n_{2}))) \} \pplus eos'}{\langle \mathsf{\{} \mathsf{(} \mathsf{*}\; n_{1} \; n_{2} \mathsf{)}\mathsf{\}}  \pplus i, k, w, o; eos \rangle \xrightarrow{{\#}(n_{1}) + \mathsf{\#}(n_{2})} \langle i, k, w, \mathsf{\{} n_{1}\mathsf{*} n_{2} \mathsf{\}} \pplus o; eos'' \rangle} \\
  \and
  \inferrule* [lab=NumMult2]{w = \mathsf{\{} \mathsf{(} \mathsf{*}\; n_{1} \; n_{2} \mathsf{)}_{\chi(p,e)}\mathsf{\}} \pplus w' \\ eos = \{(h(p)\; e)\} \pplus eos', \and eos'' = \{(h(p)\; (e - \mathsf{\#}(n_{1}) + \mathsf{\#}(n_{2}))) \} \pplus eos'}{\langle i, k, w, o; eos \rangle \xrightarrow{{\#}(n_{1}) + \mathsf{\#}(n_{2})} \langle i, k, w, \mathsf{\{} n_{1}\mathsf{*} n_{2} \mathsf{\}} \pplus o; eos'' \rangle} \\
\end{mathpar}
\begin{mathpar}
  \inferrule* [lab=StrAdd1]{((e + e') - \mathsf{\#}(s_{1}) + \mathsf{\#}(s_{2})) > 0 \\ eos = \{(h(p)\; e')\} \pplus eos', \and eos'' = \{(h(p)\; ((e + e') - \mathsf{\#}(s_{1}) + \mathsf{\#}(s_{2}))) \} \pplus eos'}{\langle \mathsf{\{} \mathsf{(} \mathsf{+}\; s_{1} \; s_{2} \mathsf{)}\mathsf{\}}  \pplus i, k, w, o; eos \rangle \xrightarrow{{\#}(s_{1}) + \mathsf{\#}(s_{2})} \langle i, k, w, \mathsf{\{} s_{1}\mathsf{+} s_{2} \mathsf{\}} \pplus o; eos'' \rangle} \\
  \and
  \inferrule* [lab=StrAdd2]{w = \mathsf{\{} \mathsf{(} \mathsf{+}\; s_{1} \; s_{2} \mathsf{)}_{\chi(p,e)}\mathsf{\}} \pplus w' \\ eos = \{(h(p)\; e)\} \pplus eos', \and eos'' = \{(h(p)\; (e - \mathsf{\#}(s_{1}) + \mathsf{\#}(s_{2}))) \} \pplus eos'}{\langle i, k, w, o; eos \rangle \xrightarrow{{\#}(s_{1}) + \mathsf{\#}(s_{2})} \langle i, k, w, \mathsf{\{} s_{1}\mathsf{+} s_{2} \mathsf{\}} \pplus o; eos'' \rangle}
\end{mathpar}

\section{Compiling MeTTa to rho}

In this section we illustrate the value of having an operational semantics by developing a compiler from MeTTa to the rho-calculus, and resource-bounded MeTTa to rholang. The essence of the translation is to use a channels for each of the registers. 

\subsection{MeTTa to the rho-calculus}
In the translation given below we employ two semantic functions. One, written $\meaningof{-}_{C}$, turns a state into a collection of data deployed at four channels, $i,k,w,o$, and another, written $\meaningof{-}_{E}$, that processes that data. We also employ a semantic function, written $\meaningof{-}$,that transliterates MeTTa term syntax into rho term syntax for use in pattern matching, etc.

The astute reader will notice that there is a question about the well-definedness of $\meaningof{-}_{E}$. It is possible that there are multiple transitions out of a single state. The actual function takes a sum of all possible transitions:

\begin{mathpar}
  \meaningof{S}_{E} := \Sigma_{r \in \{ r | S \xrightarrow{r} S' \}}\meaningof{S}_{r}
\end{mathpar}

where $\meaningof{-}_{r}$ are defined below.

The meaning of a MeTTa computation is given as the composition of the configuration and evaluation functions. That is,

\begin{mathpar}
  \meaningof{\langle i, k, w, o \rangle}_{M} = \meaningof{\langle i, k, w, o \rangle}_{C} \; \mathsf{|} \; \meaningof{\langle i, k, w, o \rangle}_{E}
\end{mathpar}

The reason for this factorization of the semantics is to facilitate the proof of correctness, as we will see in the proof. %

\subsubsection{Space configuration}

\begin{mathpar}
  \begin{array}{lll}
    \meaningof{\langle \mathsf{\{} t \mathsf{\}} \pplus i, k, w, o \rangle}_{C}(i,k,w,o) & = & i\mathsf{!}\mathsf{(}\meaningof{t}\mathsf{)}\; \mathsf{|}\; \meaningof{\langle i, k, w, o \rangle}_{C}(i,k,w,o) \\
    \meaningof{\langle i, \mathsf{\{} t \mathsf{\}} \pplus k, w, o \rangle}_{C}(i,k,w,o) & = & k\mathsf{!}\mathsf{(}\meaningof{t}\mathsf{)}\; \mathsf{|}\; \meaningof{\langle i, k, w, o \rangle}_{C}(i,k,w,o) \\
    \meaningof{\langle i, k, \mathsf{\{} t \mathsf{\}} \pplus w, o \rangle}_{C}(i,k,w,o) & = & w\mathsf{!}\mathsf{(}\meaningof{t}\mathsf{)}\; \mathsf{|}\; \meaningof{\langle i, k, w, o \rangle}_{C}(i,k,w,o) \\
    \meaningof{\langle i, k, w, \mathsf{\{} t \mathsf{\}} \pplus o \rangle}_{C}(i,k,w,o) & = & o\mathsf{!}\mathsf{(}\meaningof{t}\mathsf{)}\; \mathsf{|}\; \meaningof{\langle i, k, w, o \rangle}_{C}(i,k,w,o) \\
  \end{array}
\end{mathpar}

\subsubsection{Space evaluation}
\begin{mathpar}
  \begin{array}{lll}
    \meaningof{\langle \mathsf{\{} t' \mathsf{\}} \pplus i, \mathsf{\{} (\mathsf{=}\; t_{1} \; u_{1}), \ldots, (\mathsf{=}\; t_{n} \; u_{n}) \mathsf{\}} \pplus k, w, o \rangle}_{Query}(i,k,w,o) & & \\
    = & & \\
    \mathsf{for}( \meaningof{t'} \leftarrow i )\{ \mathsf{(} \mathsf{new}\; s\mathsf{)} \{ \Pi_{j=1}^{n} \mathsf{for}\mathsf{(} \mathsf{(} \mathsf{=} \; \meaningof{t'} \; \mathsf{(} \meaningof{t_{j}}\; v \mathsf{)} \mathsf{)} \leftarrow k \mathsf{)}\{ w\mathsf{!}(\meaningof{u_{j}}) \mathsf{|} s\mathsf{!}\mathsf{(} \mathsf{(} \mathsf{=} \; \meaningof{t_{j}} \; \mathsf{(} \meaningof{t_{j}}\; v \mathsf{)} \mathsf{)} \mathsf{)} \} & & \\
    \quad \quad \quad \quad \quad \quad \mathsf{|}\; \mathsf{for}\mathsf{(} r_{1} \leftarrow s \; \mathsf{\&}\; \ldots \; \mathsf{\&}\; r_{n} \leftarrow s \mathsf{)}\{ & & \\
    \quad \quad \quad \quad \quad \quad \quad \quad \Pi_{j=1}^{n} k\mathsf{!}\mathsf{(} r_{j} \mathsf{)}\; \mathsf{|}\; \meaningof{\langle i, \mathsf{\{} (\mathsf{=}\; t_{1} \; u_{1}), \ldots, (\mathsf{=}\; t_{n} \; u_{n}) \mathsf{\}} \pplus k, w, o \rangle}_{E}(i,k,w,o) & & \\
    \quad \quad \quad \quad \quad \quad \quad \} \} \} & &\\
    \meaningof{\langle i, \mathsf{\{} (\mathsf{=}\; t_{1} \; u_{1}), \ldots, (\mathsf{=}\; t_{n} \; u_{n}) \mathsf{\}} \pplus k, \mathsf{\{} u \mathsf{\}} \pplus w, o \rangle}_{Chain}(i,k,w,o) & & \\
    = & & \\
    \mathsf{for}( \meaningof{u} \leftarrow w )\{ \mathsf{(} \mathsf{new}\; s\mathsf{)}\{ \Pi_{j=1}^{n} \mathsf{for}\mathsf{(} \mathsf{(} \mathsf{(} \mathsf{=} \; \meaningof{u}\; \mathsf{(} \meaningof{t_{j}}\; v \mathsf{)}\mathsf{)} \leftarrow k \mathsf{)}\{ w\mathsf{!}(\meaningof{u_{j}}) \mathsf{|} s\mathsf{!}\mathsf{(} \mathsf{(} \mathsf{=} \; \meaningof{u} \; \mathsf{(} \meaningof{t_{j}}\; v \mathsf{)} \mathsf{)} \mathsf{)} \} & & \\
    \quad \quad \quad \quad \quad \quad \mathsf{|}\; \mathsf{for}\mathsf{(} r_{1} \leftarrow s \; \mathsf{\&}\; \ldots \; \mathsf{\&}\; r_{n} \leftarrow s \mathsf{)}\{ & & \\
    \quad \quad \quad \quad \quad \quad \quad \quad \Pi_{j=1}^{n} k\mathsf{!}\mathsf{(} r_{j} \mathsf{)}\; \mathsf{|}\; \meaningof{\langle i, \mathsf{\{} (\mathsf{=}\; t_{1} \; u_{1}), \ldots, (\mathsf{=}\; t_{n} \; u_{n}) \mathsf{\}} \pplus k, w, o \rangle}_{E}(i,k,w,o) & & \\
    \quad \quad \quad \quad \quad \quad \quad \} \} \} & &\\
    \meaningof{\langle \mathsf{\{} \mathsf{(}\mathsf{transform}\; t \; u\mathsf{)} \mathsf{\}} \pplus i, \mathsf{\{} t_{1}, \ldots, t_{n} \mathsf{\}} \pplus k, w, o \rangle}_{Transform}(i,k,w,o) & & \\
    = & & \\
    \mathsf{for}( \mathsf{(}\mathsf{transform}\; \meaningof{t}\; \meaningof{u}\mathsf{)} \leftarrow i )\{ \mathsf{(} \mathsf{new}\; s\mathsf{)} \{ \Pi_{j=1}^{n} \mathsf{for}\mathsf{(} \meaningof{t} \leftarrow k \mathsf{)}\{ w\mathsf{!}\mathsf{(}\meaningof{u}\mathsf{)} \mathsf{|} s \mathsf{!}\mathsf{(} \meaningof{t} \mathsf{)} \} & & \\
    \quad \quad \quad \quad \quad \quad \quad \quad \quad \quad \quad \quad \mathsf{|}\; \mathsf{for}\mathsf{(} r_{1} \leftarrow s \; \mathsf{\&}\; \ldots \; \mathsf{\&}\; r_{n} \leftarrow s \mathsf{)}\{ & & \\
    \quad \quad \quad \quad \quad \quad \quad \quad \quad \quad \quad \quad \quad \Pi_{j=1}^{n} k\mathsf{!}\mathsf{(} r_{j} \mathsf{)}\; \mathsf{|}\; \meaningof{\langle i, \mathsf{\{} t_{1}, \ldots, t_{n} \mathsf{\}} \pplus k, w, o \rangle}_{E}(i,k,w,o) & & \\
    \quad \quad \quad \quad \quad \quad \quad \quad \quad \quad \quad \quad \} \} \} & &\\
    \meaningof{\langle \mathsf{\{} \mathsf{(} \mathsf{addAtom}\; t\mathsf{)}\mathsf{\}}  \pplus i, k, w, o \rangle}_{AddAtom1}(i,k,w,o) & & \\
    = & & \\
    \mathsf{for}( \mathsf{(}\mathsf{addAtom}\; \meaningof{t} \mathsf{)} \leftarrow i )\{ k\mathsf{!}(\meaningof{t}) \;
    \mathsf{|}\; \meaningof{\langle i, k, w, o \rangle}_{E}(i,k,w,o) \} & &\\
    \meaningof{\langle i, \mathsf{\{} \mathsf{(} \mathsf{addAtom}\; t\mathsf{)}\mathsf{\}} \pplus k, w, o\rangle}_{AddAtom2}(i,k,w,o) & & \\
    = & & \\
    \mathsf{for}( \mathsf{(}\mathsf{addAtom}\; \meaningof{t} \mathsf{)} \leftarrow k )\{ k\mathsf{!}(\meaningof{t}) \;
    \mathsf{|}\; \meaningof{\langle i, \mathsf{\{} \mathsf{(} \mathsf{addAtom}\; t\mathsf{)}\mathsf{\}} \pplus k, w, o \rangle}_{E}(i,k,w,o) \} & &\\
    \meaningof{\langle \mathsf{\{} \mathsf{(} \mathsf{remAtom}\; t\mathsf{)}\mathsf{\}}  \pplus i, \mathsf{\{} t \mathsf{\}} \pplus k, w, o \rangle}_{RemAtom1}(i,k,w,o) & & \\
    = & & \\
    \mathsf{for}( \mathsf{(}\mathsf{remAtom}\; \meaningof{t} \mathsf{)} \leftarrow i )\{ \mathsf{for}\mathsf{(} \meaningof{t} \leftarrow k \mathsf{)}\{ o\mathsf{!}(\meaningof{()}) \} \;
    \mathsf{|}\; \meaningof{\langle i, k, w, o \rangle}_{E}(i,k,w,o)  \} & &\\
    \meaningof{\langle i, \mathsf{\{} \mathsf{(} \mathsf{remAtom}\; t\mathsf{)} \mathsf{\}} \pplus \mathsf{\{} t \mathsf{\}} \pplus k, w, o\rangle}_{RemAtom2}(i,k,w,o) & & \\
    = & & \\
    \mathsf{for}( \mathsf{(}\mathsf{remAtom}\; \meaningof{t} \mathsf{)} \leftarrow k )\{ \mathsf{for}\mathsf{(} \meaningof{t} \leftarrow k \mathsf{)}\{ o\mathsf{!}(\meaningof{()}) \} \;
    \mathsf{|}\; \meaningof{\langle i, \mathsf{\{} \mathsf{(} \mathsf{remAtom}\; t\mathsf{)} \mathsf{\}} \pplus k, w, o \rangle}_{E}(i,k,w,o) \} & &\\
  \end{array}
\end{mathpar}

\subsection{Correctness of the translation}
\begin{theorem}[MeTTa2rho correctness]
  \begin{mathpar}
    S_{1} \wbbisim S_{2} \iff \meaningof{S_{1}}_{M} \wbbisim \meaningof{S_{2}}_{M}
  \end{mathpar}
\end{theorem}
\begin{proof}
  Proof sketch: Essentially, the translation is correct-by-construction. Intuitively, we see that there is a bisimulation relation between MeTTa computations and their translations in the rho-calculus. This bisimulation may be composed with a bisimulation between MeTTa calculations to yield a bisimulation in the rho-translation, and vice versa.
  The bisimulation bridging between the two domains is effectively represented in the translation. For each different kind of state, the left hand side of each bisimulation pair is the left hand side of the \emph{definition} evaluation semantic function and the right hand side of the bisimulation pair is the right hand side of the definition. Hence, correct-by-construction.
  The formal proof uses terms in the input, working, and output registers as barbs for the notion of bisimulation on MeTTa computations, while their translations via the configuration function are barbs in the bisimulation in the rho-calculus.
\end{proof}

\subsection{Resource-bounded MeTTa to rholang}
\subsubsection{Space configuration}
\begin{mathpar}
  \begin{array}{lll}
    \meaningof{\langle \mathsf{\{} t \mathsf{\}} \pplus i, k, w, o; eos \rangle}(i,k,w,o,c) & = & i\mathsf{!}\mathsf{(}\meaningof{t}\mathsf{)}\; \mathsf{|}\; \meaningof{\langle i, k, w, o; eos \rangle}(i,k,w,o,c) \\
    \meaningof{\langle i, \mathsf{\{} t \mathsf{\}} \pplus k, w, o; eos \rangle}(i,k,w,o,c) & = & k\mathsf{!}\mathsf{(}\meaningof{t}\mathsf{)}\; \mathsf{|}\; \meaningof{\langle i, k, w, o; eos \rangle}(i,k,w,o,c) \\
    \meaningof{\langle i, k, \mathsf{\{} t \mathsf{\}} \pplus w, o; eos \rangle}(i,k,w,o,c) & = & w\mathsf{!}\mathsf{(}\meaningof{t}\mathsf{)}\; \mathsf{|}\; \meaningof{\langle i, k, w, o; eos \rangle}(i,k,w,o,c) \\
    \meaningof{\langle i, k, w, \mathsf{\{} t \mathsf{\}} \pplus o; eos \rangle}(i,k,w,o,c) & = & o\mathsf{!}\mathsf{(}\meaningof{t}\mathsf{)}\; \mathsf{|}\; \meaningof{\langle i, k, w, o; eos \rangle}(i,k,w,o,c) \\
    \meaningof{\langle i, k, w, o; \mathsf{\{} \mathsf{(}h(p) e \mathsf{)} \mathsf{\}} \pplus eos \rangle}(i,k,w,o,c) & = & c\mathsf{!}\mathsf{(}\meaningof{h(p)} \meaningof{e}\mathsf{)}\; \mathsf{|}\; \meaningof{\langle i, k, w, o; eos \rangle}(i,k,w,o,c) \\
  \end{array}
\end{mathpar}

\subsubsection{Space evaluation}
\begin{mathpar}
  \begin{array}{lll}
    \meaningof{\langle \mathsf{\{} t'_{\chi(p,e)} \mathsf{\}} \pplus i, \mathsf{\{} (\mathsf{=}\; t_{1} \; u_{1}), \ldots, (\mathsf{=}\; t_{n} \; u_{n}) \mathsf{\}} \pplus k, w, o; eos \rangle}_{Query}(i,k,w,o,c) & & \\
    = & & \\
    \mathsf{for}( \meaningof{t'}_{\chi(p,e)} \leftarrow i \;\mathsf{\&}\; \mathsf{(} \meaningof{h(p)}\;e'\mathsf{)} \leftarrow c )\{ & & \\
    \quad \mathsf{if} \mathsf{(} \mathsf{(} \mathsf{(} \meaningof{e} \mathsf{+} \meaningof{e'}\mathsf{)} \mathsf{-} \mathsf{(} \Sigma_{i} \meaningof{\mathsf{\#}(\sigma_{i})} + \Sigma_{i} \meaningof{\mathsf{\#}(u_{i}\sigma_{i})} \mathsf{)} \mathsf{)} \mathsf{>} \mathsf{0}\mathsf{)} \{ & & \\
    \quad \quad \quad \mathsf{(} \mathsf{new}\; s\mathsf{)} \{ \Pi_{j=1}^{n} \mathsf{for}\mathsf{(} \mathsf{(} \mathsf{=} \; \meaningof{t'} \; v \mathsf{)} \leftarrow k \mathsf{)}\{ w\mathsf{!}(\meaningof{u_{j}}) \mathsf{|} s\mathsf{!}\mathsf{(} \mathsf{(} \mathsf{=} \; \meaningof{t'} \; v \mathsf{)} \mathsf{)} \} & & \\
    \quad \quad \quad \quad \quad \quad \quad \mathsf{|}\; \mathsf{for}\mathsf{(} r_{1} \leftarrow s \; \mathsf{\&}\; \ldots \; \mathsf{\&}\; r_{n} \leftarrow s \mathsf{)}\{ & & \\
    \quad \quad \quad \quad \quad \quad \quad \quad \quad \Pi_{j=1}^{n} k\mathsf{!}\mathsf{(} r_{j} \mathsf{)}\; \mathsf{|}\; \meaningof{\langle i, \mathsf{\{} (\mathsf{=}\; t_{1} \; u_{1}), \ldots, (\mathsf{=}\; t_{n} \; u_{n}) \mathsf{\}} \pplus k, w, o; eos \rangle}_{E}(i,k,w,o,c) & & \\
    \quad \quad \quad \quad \quad \quad \quad \quad \} \} \} & & \\
    \quad \mathsf{else} \{ i!\mathsf{(} \meaningof{t'}_{\chi(p,e)} \mathsf{)}\; \mathsf{|}\; c\mathsf{!}\mathsf{(} \mathsf{(} \meaningof{h(p)}\;e'\mathsf{)} \mathsf{)} \}
    \} & &\\
    \meaningof{\langle i, \mathsf{\{} (\mathsf{=}\; t_{1} \; u_{1}), \ldots, (\mathsf{=}\; t_{n} \; u_{n}) \mathsf{\}} \pplus k, \mathsf{\{} u_{\chi(p,\bot)} \mathsf{\}} \pplus w, o; eos \rangle}_{Chain}(i,k,w,o,c) & & \\
    = & & \\
    \mathsf{for}( \meaningof{u}_{\chi(p,\bot)} \leftarrow o \;\mathsf{\&}\; \mathsf{(} \meaningof{h(p)}\;e\mathsf{)} \leftarrow c)\{ & & \\
    \quad \mathsf{if} \mathsf{(} \mathsf{(} \meaningof{e} \mathsf{-} \mathsf{(} \Sigma_{i} \meaningof{\mathsf{\#}(\sigma_{i})} + \Sigma_{i} \meaningof{\mathsf{\#}(u_{i}\sigma_{i})} \mathsf{)} \mathsf{)} \mathsf{>} \mathsf{0}\mathsf{)} \{ & & \\
    \quad \quad \quad \mathsf{(} \mathsf{new}\; s\mathsf{)}\{ \Pi_{j=1}^{n} \mathsf{for}\mathsf{(} \mathsf{(} \mathsf{(} \mathsf{=} \; \meaningof{u}\; v \mathsf{)} \leftarrow k \mathsf{)}\{ w\mathsf{!}(\meaningof{u_{j}}) \mathsf{|} s\mathsf{!}\mathsf{(} \mathsf{(} \mathsf{=} \; \meaningof{u} \; v \mathsf{)} \mathsf{)} \} & & \\
    \quad \quad \quad \quad \quad \quad \quad \mathsf{|}\; \mathsf{for}\mathsf{(} r_{1} \leftarrow s \; \mathsf{\&}\; \ldots \; \mathsf{\&}\; r_{n} \leftarrow s \mathsf{)}\{ & & \\
    \quad \quad \quad \quad \quad \quad \quad \quad \quad \Pi_{j=1}^{n} k\mathsf{!}\mathsf{(} r_{j} \mathsf{)}\; \mathsf{|}\; \meaningof{\langle i, \mathsf{\{} (\mathsf{=}\; t_{1} \; u_{1}), \ldots, (\mathsf{=}\; t_{n} \; u_{n}) \mathsf{\}} \pplus k, w, o; eos \rangle}_{E}(i,k,w,o,c) & & \\
    \quad \quad \quad \quad \quad \quad \quad \quad \} \} \} & & \\
    \quad \mathsf{else} \{ i!\mathsf{(} \meaningof{u}_{\chi(p,\bot)} \mathsf{)}\; \mathsf{|}\; c\mathsf{!}\mathsf{(} \mathsf{(} \meaningof{h(p)}\;e\mathsf{)} \mathsf{)} \}
    \} & &\\
    \meaningof{\langle \mathsf{\{} \mathsf{(}\mathsf{transform}\; t \; u\mathsf{)}_{\chi(p,e)} \mathsf{\}} \pplus i, \mathsf{\{} t_{1}, \ldots, t_{n} \mathsf{\}} \pplus k, w, o; eos \rangle}_{transform}(i,k,w,o,c) & & \\
    = & & \\
    \mathsf{for}( \mathsf{(}\mathsf{transform}\; \meaningof{t}\; \meaningof{u}\mathsf{)}_{\chi(p,e)} \leftarrow i \;\mathsf{\&}\; \mathsf{(} \meaningof{h(p)}\;e'\mathsf{)} \leftarrow c)\{ & & \\
    \quad \mathsf{if} \mathsf{(} \mathsf{(} \mathsf{(} \meaningof{e} \mathsf{+} \meaningof{e'}\mathsf{)} \mathsf{-} \mathsf{(} \Sigma_{i} \meaningof{\mathsf{\#}(\sigma_{i})} + \Sigma_{i} \meaningof{\mathsf{\#}(u\sigma_{i})} \mathsf{)} \mathsf{)} \mathsf{>} \mathsf{0}\mathsf{)} \{ & & \\
    \quad \quad \quad \mathsf{(} \mathsf{new}\; s\mathsf{)} \{ \Pi_{j=1}^{n} \mathsf{for}\mathsf{(} \meaningof{t} \leftarrow k \mathsf{)}\{ w\mathsf{!}\mathsf{(}\meaningof{u}\mathsf{)} \mathsf{|} s \mathsf{!}\mathsf{(} \meaningof{t} \mathsf{)} \} & & \\
    \quad \quad \quad \quad \quad \quad \quad \mathsf{|}\; \mathsf{for}\mathsf{(} r_{1} \leftarrow s \; \mathsf{\&}\; \ldots \; \mathsf{\&}\; r_{n} \leftarrow s \mathsf{)}\{ & & \\
    \quad \quad \quad \quad \quad \quad \quad \quad \Pi_{j=1}^{n} k\mathsf{!}\mathsf{(} r_{j} \mathsf{)}\; \mathsf{|}\; \meaningof{\langle i, \mathsf{\{} t_{1}, \ldots, t_{n} \mathsf{\}} \pplus k, w, o; eos \rangle}_{E}(i,k,w,o,c) & & \\
    \quad \quad \quad \quad \quad \quad \quad \} \} \} & &\\
    \quad \mathsf{else} \{ i!\mathsf{(}\mathsf{(}\mathsf{transform}\; \meaningof{t}\; \meaningof{u}\mathsf{)}_{\chi(p,e)}\mathsf{)}\; \mathsf{|}\; c\mathsf{!}\mathsf{(} \mathsf{(} \meaningof{h(p)}\;e'\mathsf{)} \mathsf{)} \}
    \} & &
  \end{array}
\end{mathpar}
\begin{mathpar}
  \begin{array}{lll}
    \meaningof{\langle \mathsf{\{} \mathsf{(} \mathsf{addAtom}\; t\mathsf{)}_{\chi(p,e)}\mathsf{\}}  \pplus i, k, w, o; eos \rangle}_{E}(i,k,w,o,c) & & \\
    = & & \\
    \mathsf{for}( \mathsf{(}\mathsf{addAtom}\; \meaningof{t} \mathsf{)}_{\chi(p,e)} \leftarrow i \;\mathsf{\&}\; \mathsf{(} \meaningof{h(p)}\;e'\mathsf{)} \leftarrow c)\{ & & \\
    \quad \mathsf{if} \mathsf{(} \mathsf{(} \mathsf{(} \meaningof{e} \mathsf{+} \meaningof{e'}\mathsf{)} \mathsf{-} \meaningof{\mathsf{\#}(t)} \mathsf{)} \mathsf{>} \mathsf{0}\mathsf{)} \{ & & \\
    \quad\quad \quad k\mathsf{!}(\meaningof{t}) \;
    \mathsf{|}\; \meaningof{\langle i, k, w, o; eos \rangle}_{E}(i,k,w,o,c) \} & &\\
    \quad \mathsf{else} \{ i!\mathsf{(}\mathsf{(}\mathsf{addAtom}\; \meaningof{t}\mathsf{)}_{\chi(p,e)}\mathsf{)}\; \mathsf{|}\; c\mathsf{!}\mathsf{(} \mathsf{(} \meaningof{h(p)}\;e'\mathsf{)} \mathsf{)} \}
    \} & & \\
    \meaningof{\langle i, \mathsf{\{} \mathsf{(} \mathsf{addAtom}\; t\mathsf{)}\mathsf{\}} \pplus k, w, o\rangle}_{AddAtom2}(i,k,w,o,c) & & \\
    = & & \\
    \mathsf{for}( \mathsf{(}\mathsf{addAtom}\; \meaningof{t} \mathsf{)} \leftarrow k )\{ k\mathsf{!}(\meaningof{t}) \;
    \mathsf{|}\; \meaningof{\langle i, \mathsf{\{} \mathsf{(} \mathsf{addAtom}\; t\mathsf{)}\mathsf{\}} \pplus k, w, o; eos \rangle}_{E}(i,k,w,o,c) \} & &\\
    \meaningof{\langle \mathsf{\{} \mathsf{(} \mathsf{remAtom}\; t\mathsf{)}_{\chi(p,e)}\mathsf{\}}  \pplus i, \mathsf{\{} t \mathsf{\}} \pplus k, w, o; eos \rangle}_{RemAtom1}(i,k,w,o,c) & & \\
    = & & \\
    \mathsf{for}( \mathsf{(}\mathsf{remAtom}\; \meaningof{t} \mathsf{)}_{\chi(p,e)} \leftarrow i \;\mathsf{\&}\; \mathsf{(} \meaningof{h(p)}\;e'\mathsf{)} \leftarrow c)\{ & & \\
    \quad \mathsf{if} \mathsf{(} \mathsf{(} \mathsf{(} \meaningof{e} \mathsf{+} \meaningof{e'}\mathsf{)} \mathsf{-} \meaningof{\mathsf{\#}(t)} \mathsf{)} \mathsf{>} \mathsf{0}\mathsf{)} \{ & & \\
    \quad \quad \mathsf{for}\mathsf{(} \meaningof{t} \leftarrow k \mathsf{)}\{ o\mathsf{!}(\meaningof{()}) \} \;
    \mathsf{|}\; \meaningof{\langle i, k, w, o; eos \rangle}_{E}(i,k,w,o,c)  \} & &\\
    \quad \mathsf{else} \{ i!\mathsf{(}\mathsf{(}\mathsf{remAtom}\; \meaningof{t}\mathsf{)}_{\chi(p,e)}\mathsf{)}\; \mathsf{|}\; c\mathsf{!}\mathsf{(} \mathsf{(} \meaningof{h(p)}\;e'\mathsf{)} \mathsf{)} \}
    \} & & \\
    \meaningof{\langle i, \mathsf{\{} \mathsf{(} \mathsf{remAtom}\; t\mathsf{)} \mathsf{\}} \pplus \mathsf{\{} t \mathsf{\}} \pplus k, w, o\rangle}_{RemAtom2}(i,k,w,o,c) & & \\
    = & & \\
    \mathsf{for}( \mathsf{(}\mathsf{remAtom}\; \meaningof{t} \mathsf{)} \leftarrow k )\{ \mathsf{for}\mathsf{(} \meaningof{t} \leftarrow k \mathsf{)}\{ o\mathsf{!}(\meaningof{()}) \} \;
    \mathsf{|}\; \meaningof{\langle i, \mathsf{\{} \mathsf{(} \mathsf{remAtom}\; t\mathsf{)} \mathsf{\}} \pplus k, w, o; eos \rangle}_{E}(i,k,w,o,c) \} & &\\
  \end{array}
\end{mathpar}

\subsection{Correctness of the translation}
\begin{theorem}[Resource-bounded MeTTa2rho correctness]
  \begin{mathpar}
    S_{1} \wbbisim S_{2} \iff \meaningof{S_{1}}_{M} \wbbisim \meaningof{S_{2}}_{M}
  \end{mathpar}
\end{theorem}

\begin{proof}
  Proof sketch: while all the resource accounting adds to the complexity of the translation, the proof essentially reprises the previous one. 
\end{proof}

\section{Conclusion and future work}

We have presented two versions of the operational semantics for MeTTa, one that is fit for private implementations that have some external security model, and one that is fit for running in a decentralized setting.

This semantics does not address typed versions of MeTTa. An interesting avenue of approach is to apply Meredith and Stay's OSLF to this semantics to derive a type system for MeTTa that includes spatial and behavioral types \cite{DBLP:journals/corr/abs-2102-04672}.

\bibliographystyle{plain}   
\bibliography{mops.bib}

\begin{thebibliography}{10}

\bibitem{DBLP:conf/sac/Buday15}
Gergely Buday.
\newblock Formalising the {SECD} machine with nominal isabelle.
\newblock In {\em {SAC}}, pages 1823--1824. {ACM}, 2015.

\bibitem{YT:GoertzelIklePotapovHyperon2022}
Ben Goertzel, Matt Ikle, and Alexey Potapov.
\newblock Hyperon as agi approach, Nov 2022.

\bibitem{DBLP:conf/agi/HartG08}
David Hart and Ben Goertzel.
\newblock Opencog: {A} software framework for integrative artificial general
  intelligence.
\newblock In {\em {AGI}}, volume 171 of {\em Frontiers in Artificial
  Intelligence and Applications}, pages 468--472. {IOS} Press, 2008.

\bibitem{DBLP:conf/icfp/KiselyovSFS05}
Oleg Kiselyov, Chung{-}chieh Shan, Daniel~P. Friedman, and Amr Sabry.
\newblock Backtracking, interleaving, and terminating monad transformers:
  (functional pearl).
\newblock In {\em {ICFP}}, pages 192--203. {ACM}, 2005.

\bibitem{DBLP:conf/concur/LeiferM00}
James~J. Leifer and Robin Milner.
\newblock Deriving bisimulation congruences for reactive systems.
\newblock In Catuscia Palamidessi, editor, {\em {CONCUR} 2000 - Concurrency
  Theory, 11th International Conference, University Park, PA, USA, August
  22-25, 2000, Proceedings}, volume 1877 of {\em Lecture Notes in Computer
  Science}, pages 243--258. Springer, 2000.

\bibitem{JVM7Spec}
Tim Lindholm, Frank Yellin, Gilad Bracha, and Alex Buckley.
\newblock {\em The Java™ Virtual Machine Specification}.
\newblock Oracle America, Inc., java se 7 edition edition, February 2012.

\bibitem{DBLP:journals/entcs/MeredithR05}
L.~Gregory Meredith and Matthias Radestock.
\newblock A reflective higher-order calculus.
\newblock {\em Electr. Notes Theor. Comput. Sci.}, 141(5):49--67, 2005.

\bibitem{DBLP:journals/mscs/Milner92}
Robin Milner.
\newblock Functions as processes.
\newblock {\em Mathematical Structures in Computer Science}, 2(2):119--141,
  1992.

\bibitem{Plotkin04theorigins}
Gordon~D. Plotkin.
\newblock The origins of structural operational semantics.
\newblock In {\em Journal of Logic and Algebraic Programming}, pages 60--61,
  2004.

\bibitem{DBLP:books/daglib/0004377}
Davide Sangiorgi and David Walker.
\newblock {\em The Pi-Calculus - a theory of mobile processes}.
\newblock Cambridge University Press, 2001.

\bibitem{DBLP:journals/corr/abs-2102-04672}
Christian Williams and Michael Stay.
\newblock Native type theory.
\newblock {\em CoRR}, abs/2102.04672, 2021.

\bibitem{wood2014ethereum}
Gavin Wood.
\newblock Ethereum: A secure decentralised generalised transaction ledger.
\newblock {\em GitHUB}, 2014.

\end{thebibliography}

\end{document}